\documentclass[a4paper,onecolumn,11pt]{quantumarticle}
\pdfoutput=1

\usepackage{color}
\usepackage{bm}
\usepackage{amsmath, amssymb, amsthm, amsfonts, latexsym}
\usepackage[final]{hyperref}
\usepackage{cleveref}
\usepackage{enumitem}
\usepackage{caption}
\usepackage{tikz}
\usepackage{pgfplots}
\usepackage{subcaption}
\usepackage{multirow}
\usepackage{floatrow}
\usepackage{ragged2e}
\usepackage{wrapfig}
\newcommand{\reg}[1]{{\textsf{#1}}}

  \newtheorem{theorem}{Theorem}
  \newtheorem{lemma}[theorem]{Lemma}
  \newtheorem{definition}[theorem]{Definition}
\newtheorem{corollary}[theorem]{Corollary}
\newtheorem{remark}[theorem]{Remark}

\newcommand{\tit}{A simple protocol for verifiable delegation of quantum computation in one round}
\title{\tit} 

\author{Alex B. Grilo}
\affiliation{CWI and QuSoft, Amsterdam, The
Netherlands}

\newcommand{\beq}{\begin{eqnarray}}
\newcommand{\eeq}{\end{eqnarray}}

\newcommand{\ket}[1]{|#1\rangle}
\newcommand{\bra}[1]{\langle#1|}
\newcommand{\kb}[1]{|#1\rangle\langle#1|}

\newcommand {\Tr} {\ensuremath{ \mathrm{Tr} }}

\newcommand {\trace} [1] {\fn{\Tr}{#1}}

\newcommand{\Es}[1]{\ensuremath{\mathop{\textsc{E}}}_{#1}}
\newcommand {\eqdef} {\ensuremath{ \stackrel{\mathrm{def}}{=} }}

\newcommand{\uX}{{\sf X}}
\newcommand{\uI}{{\sf I}}
\newcommand{\uZ}{{\sf Z}}

\newcommand{\oX}{\sigma_X}
\newcommand{\oI}{\sigma_I}
\newcommand{\oZ}{\sigma_Z}

\newcommand{\setft}[1]{\mathrm{#1}}

\newcommand{\Obs}{\setft{Obs}}

\DeclareMathOperator{\poly}{poly}

\newcommand{\real}{\ensuremath{\mathbb{R}}}
\newcommand{\complex}{\ensuremath{\mathbb{C}}}

\newcommand{\T}{\ensuremath{\mathcal{T}}}

\newcommand{\mH}{\mathcal{H}}

\newcommand{\eps}{\varepsilon}

\newcommand{\EPR}{\Phi_{00}}

\newcommand{\aux}{\textsc{aux}}

\newcommand{\energy}{\textsc{ET}}
\newcommand{\et}{\energy{}}
\newcommand{\pbt}{\textsc{PBT}}

\newcommand{\ver}{\textsc{V}}

\newcommand{\inr}{\in_R}

\newcommand {\minusspace} {\: \! \!}

\newcommand{\class}[1]{\ensuremath{\mathsf{#1}}}

\newcommand{\QMA}{\class{QMA}}

\newcommand{\Zplus}{\ensuremath{\mathbb{Z}^+}}
\newcommand{\fn}[2]{\ensuremath{ #1 \minusspace \br{ #2 } }}

\newcommand{\br}[1]{\ensuremath{ \left( #1 \right) }}

\newcommand{\binary}{\{0,1\}}

\newcommand\norm[1]{\left\lVert#1\right\rVert}

\newcommand\defclass[5]{
\begin{definition}[#1]\label{#2}
#3
\begin{description}
  \item[Completeness.] #4
  \item[Soundness.] #5
\end{description} 
\end{definition} 
}

\begin{document}

\maketitle

\begin{abstract}
  The importance of being able to verify quantum computation delegated to
  remote servers increases with recent development of quantum technologies.
  In some of the proposed protocols for this task, a client delegates her 
  quantum computation to  non-communicating servers in multiple rounds of communication. In this work, we propose the first protocol where the client delegates her quantum computation to two servers in one-round of communication. 
  Another advantage of our protocol is that it is conceptually simpler than previous protocols. The parameters of our protocol also make it possible to prove security even if the servers are
  allowed to communicate, but respecting the plausible assumption that information cannot be
  propagated faster than speed of light, making it the first relativistic protocol for quantum computation.
\end{abstract}

\section{Introduction}
With the recent progress in the development of quantum technologies, large-scale
quantum computers may be available in a not-so-distant future.  Their costs and
infrastructure requirements make it impractical for them to be ubiquitous,
however clients could send their quantum computation to be performed remotely by a quantum
server in the cloud~\cite{Castelvecchi17}, broadening the use of  quantum
advantage to solve computational problems (see Ref. \cite{Montanaro16} for such
examples).  For the clients, it is a major concern whether the quantum servers
are performing the correct computation and quantum speedup is really being
experienced.

In order to solve this problem, we aim a protocol for verifiable delegation of
quantum computation where the client exchanges messages with the server, and, at
the end of the protocol, either the client holds the output of her computation,
or she detects that the server is defective.  Ideally, the client is a classical
computer and an honest server only needs polynomial-time quantum computation to
answer correctly. Protocols of this form could also be used for validating
devices that claim to have quantum computational power, but in this work we
focus on the point of view of verifiable delegation of computation.

There are efficient protocols that can perform this task if the model is
relaxed, for instance giving limited quantum power and quantum communication to
the client~\cite{FitzsimonsK12,AharonovBOEM17,Broadbent15,Morimae14,MorimaeF16} or if the quantum server is polynomially bounded~\cite{Mahadev18,ACGH20}.
In this work, we focus on a second line of protocols, where a classical client
delegates her computation to non-communicating quantum servers. Although the
servers are supposed to share and maintain entangled states, which is feasible
in principle but technologically challenging, these protocols are
``plug-and-play'' in the sense that the client only needs classical
communication with the quantum servers.

Following standard notation in these protocols, we start calling the client and
servers by verifier and provers, respectively.  The security of such protocols
relies on the so called self-testing of non-local games.  We consider games
where a verifier interacts with non-communicating provers by exchanging one
round of classical communication and, based on the correlation of the provers'
answers, the verifier decides to accept or reject. The goal of the provers is to
maximize the acceptance probability in the game and they can share a common
strategy before the game starts.  A game is non-local~\cite{Bell64} whenever
there exists a quantum strategy for the provers that achieves acceptance
probability strictly higher than any classical strategy, allowing the verifier
to certify that the provers share some entanglement, if the classical bound is
surpassed.  Self-testing~\cite{MayersY04} goes one step further, proving that if
the correlation of the provers' answers is close to the optimal quantum value,
their strategy  is close to the honest one.

Reichardt, Unger and Vazirani~\cite{ReichardtUV12} used the ideas of
self-testing to propose a verifiable delegation scheme where the verifier
interleaves questions of non-local games and instructions for the computation,
and from the point of view of the provers, these two types of questions  are
indistinguishable. In this case, the correctness of the quantum computation is
inherited by the guarantees achieved in self-testing.  Follow-up works
\cite{McKague16,GheorghiuKW15,HajdusekPDF15,FitzsimonsH15,NatarajanV16,ColadangeloGJV17}
have used the same approach  in order to propose more efficient protocols.

In this work, we present the first one-round protocol for verifiable
delegation of quantum computation.
We also notice that the protocol is conceptually simple, in
contrast with previous protocols that have a rather complicated structure. We
expect that its main ideas can be generalized to other contexts as MIP$^*$
protocols for iterated non-deterministic exponential time and even in new
protocols for delegation of quantum computation.  We also remark that
the parameters of the protocol allow us to replace the unjustified assumption
that the provers do not communicate to a more plausible assumption that the
communication cannot be faster than speed of light. 

Technically, we achieve our protocol by showing a non-local game for Local Hamiltonian
problem, where the verifier plays against two provers in one round of classical
communication. In this game, honest provers perform polynomial time quantum
computation on copies of the groundstate of the Hamiltonian. This non-local game
is of independent interest since it was an open question if a one-round game for
Local Hamiltonian problem could be achieved with only two efficient provers.
This non-local game can be used as a delegation protocol through the circuit-to-hamiltonian construction. \\

\subsection{Our contributions}

\noindent\textbf{New Non-local game for Local Hamiltonians.}
The main technical result of this work is presenting
one-round two-prover game for the Local Hamiltonian problem, where honest
provers only need quantum polynomial time computation, copies of the
groundstate of the Hamiltonian and shared EPR pairs. More concretely,
we show
how to construct a game $G(H)$ based on a $XZ$ Local Hamiltonian\footnote{An $XZ$ Local Hamiltonian is a Hamiltonian that can be decomposed in sum of polynomially many terms that are  tensor
products of Paulis  $\sigma_X$, $\sigma_Z$ and $\sigma_I$}  $H$ acting on $n$
qubits and
the upper and lower bounds on the maximum acceptance probability in $G(H)$
are tightly related to the groundstate energy of $H$.
Then, based on
$G(H)$, we devise a game $\tilde{G}(H)$ such that if the groundstate energy of
$H$ is low, then the maximum
acceptance probability in $\tilde{G}(H)$ is at least $\frac{1}{2} + \Delta$, 
while if the groundstate energy is
high, the acceptance probability in the game is at most
$\frac{1}{2} - \Delta$.
We describe now the main ideas of $G(H)$.

The game is composed by two tests: the Pauli Braiding Test
(\pbt)~\cite{NatarajanV16}, where
the verifier checks if the provers share the expected state and perform
the indicated Pauli measurements, and the Energy Test (\energy{}), where the
verifier estimates the groundstate energy of $H$.

The same structure was used in a different way in the non-local game  for LH
proposed by Natarajan and Vidick~\cite{NatarajanV16} (and implicitly in
Ji~\cite{Ji16}).  In their game, $7$ provers are expected to share the encoding
of the groundstate of $H$ under a quantum error correcting code.  In
\et{}, the provers estimate the groundstate energy by jointly performing the
measurements on the state, while \pbt{}  checks if the provers share a correct
encoding of some state and if they perform the indicated measurements.  The
provers receive questions consisting in a Pauli tensor product observable and
they answer with the one-bit outcome of the measurement on their share of the
state.  The need of $7$ provers comes from the fact that the verifier must test
if the provers are committed to an encoded state and use it in all of their
measurements. It is an open problem if the number of provers can be decreased in this
setup.

In this work, we are able to reduce the number of provers to $2$ by making  them 
asymmetric.  In \et{}, one of the provers holds the groundstate of $H$
and teleports it to the second prover, who is responsible for measuring it.  In
our case, \pbt{} checks if the provers share EPR pairs and if the second
prover's measurements are correct. We remark that no test is needed for the
state, since the chosen measurement is not known by the first prover.
We notice that the size of the answers in our protocol is 
polynomial in $n$, since the verifier needs the teleportation results for every
qubit (in order to hide the measurement). We leave as an open problem if
the size of the answers can be reduced, hopefully achieving constant-size
answers as in NV.
i
We state now the key ideas to upper bound
the maximum acceptance probability  of $G(H)$.  The
behavior of the second prover in $\et{}$ can be verified thanks to $\pbt{}$,
since the two tests are indistinguishable to him.  On the other hand, the first
prover can perfectly distinguish \pbt{} and \energy{}, but he has no information
about the measurement being performed. We show that his optimal strategy is to
teleport the groundstate of $H$, but in this case the acceptance probability is
high iff the groundstate energy is low. 
~\\

\noindent\textbf{Protocol for verifiable delegation of quantum computation.}
The task of verifiable delegation of quantum computation can be easily reduced to estimating the groundenergy of local Hamiltonians through the circuit-to-hamiltonian
construction~\cite{FitzsimonsH15,NatarajanV16}, which has been called {\em post-hoc verification of quantum computation}.  In this construction, a quantum circuit $Q$ is reduced to an instance $H_Q$ of LH, such that $H_Q$
has low groundstate energy  iff $Q$ accepts with high probability. Our non-local game for $H_Q$ can be seen as a delegation protocol, where
the verifier interacts with two non-communicating entangled provers in
one-round of classical communication.

When compared to previous protocols, our result has some very nice properties. First, differently to previous results, our protocol is very simple to state, which could make it easier to be extended to other settings. 
Secondly, using standard techniques in relativistic cryptography, we can replace the unjustified assumption that the two servers do not communicate by the No Superluminal Signaling (NSS) principle: the security of the protocol would only rely that the two servers cannot communicate faster than the speed of light.

The circuit-to-hamiltonian construction also causes an overhead on the resources
needed by honest provers. Namely, in our protocol the provers need $\tilde{O}(n
g^2)$ EPR pairs for delegating the computation of a quantum circuit acting on
$n$ qubits and composed by $g$ gates, while other protocols need
only $\tilde{O}(g)$ EPR pairs~\cite{ColadangeloGJV17}. We leave as an open
problem finding more efficient two-provers one-round protocol for delegating quantum computation.

We also leave as an
open question if it is possible to create a one-round and {\em blind} verifiable
delegation scheme for quantum computation, or proving that this is improbable,
in the lines of Ref. \cite{AaronsonCGK17}.

\subsection*{Organization}

 In
\Cref{sec:preliminaries}, we give the necessary preliminaries, including the definition of the Pauli Braiding Test. Then, in \Cref{sec:game-lh} we present our non-local game for local Hamiltonian problem.

For completeness, we present in \Cref{sec:applications,sec:circ-to-ham} proof of standard theorems in Relativistic cryptography and Hamiltonian complexity that allows us to conclude with the protocol for verifiable delegation of quantum computation.

\subsection*{Acknowledgements}
I thank Iordanis Kerenidis,  Dami\'{a}n Pital\'{u}a-Garc\'{i}a and Thomas Vidick for useful
discussions and comments in early drafts of this manuscript. I also thank anonymous reviewers that pointed me how to improve the presentation of this work.
Part of this work was done when I was member of IRIF, Universit\'{e} Paris Diderot, Paris, France, where I was supported by ERC QCC and  French Programme d’Investissement
d’Avenir RISQ P141580. Partially supported by Supported by ERC Consolidator Grant 615307-QPROGRESS.

\section{Preliminaries}
\label{sec:preliminaries}
We assume basic knowledge on Quantum Computation topics and 
we refer the readers that are not familiar with them to Ref. \cite{NielsenC2011}.

\subsection{Notation}
We denote $[n]$ as the set $\{1,...,n\}$.
For a finite set $S$, we denote $x \in_R S$ as $x$ being an uniformly random
element from $S$.
We assume that all Hilbert spaces are finite-dimensional. 
For a Hilbert space $\mH$ and a  linear
operator $M$ on $\mH$, we denote $\lambda_0(M)$  as its smallest eigenvalue and $\norm{M}$ as its maximum singular value.
An $n$-qubit binary observable $O$ is a  Hermitian matrix with 
eigenvalues $\pm 1$.  We denote $\Obs(\mH)$ as the set of binary
observables on the Hilbert space $\mH$.

We will use the letters X, Z and I to denote questions in multi-prover games,  the letters in the sans-serif font $\uX$, $\uZ$ and $\uI$ to denote unitaries and $\sigma_X$, $\sigma_Z$ and $\sigma_I$ to denote observables such that 
\begin{equation*}\label{eq:pauli-matrix}
\uI = \sigma_I = \begin{pmatrix} 1 & 0 \\ 0 & 1 \end{pmatrix},\qquad \uX =
  \sigma_X = \begin{pmatrix} 0 & 1 \\ 1 & 0 \end{pmatrix}\text{and}\qquad \uZ =
    \sigma_Z = \begin{pmatrix} 1 & 0 \\ 0 & -1\end{pmatrix}.
\end{equation*}

\subsection{Non-local games, Self-testing and the Pauli Braiding Test}
\label{sec:pbt}

We consider games where a verifier plays against two
provers in the following way. The verifier sends questions to the
provers according to a publicly known distribution and the provers answer
back to the verifier. Based on the correlation of the answers, 
the verifier decides to accept or reject according to an acceptance rule that is also publicly known.
The provers share a common strategy before the
game starts in order to maximize the acceptance probability in the game, but
they do not communicate afterwards.

For a game $G$, its classical value $\omega(G)$ is the maximum acceptance
probability in the game if the provers share classical randomness, while the quantum value $\omega^*(G)$ is the maximum acceptance probability if they are 
allowed to follow a quantum strategy,
i.e. share a quantum state and apply measurements on it.
Non-local games (or Bell tests)~\cite{Bell64} are such games where $\omega^*(G) > \omega(G)$ and they have played a major role in Quantum Information Theory, since
they allow the
verifier to certify that there exists some quantumness in the strategy of the
provers, if the classical bound is surpassed.

Self-testing (also known as device-independent certification or rigidity
theorems) of a non-local game $G$ allows us to achieve stronger conclusions by 
showing that if the acceptance probability on $G$ is close to $\omega^*(G)$,
then the strategy of the provers is close to the ideal one, 
up to local isometries.

\subsubsection{Magic Square game}
\label{sec:magic-square}

The Magic Square or Mermin-Peres game~\cite{Mermin90,Peres90},
is a two-prover non-local game where one of the provers is asked a row $r \in
\{1,2,3\}$ and the second prover is asked with a column $c \in \{1,2,3\}$.
The first and second prover answer with $a_1, a_2 \in \{\pm 1\}$ and 
$b_1, b_2 \in \{\pm 1\}$, respectively.
By setting  $a_3 = a_1 \oplus a_2$ and $b_3 = b_1 \oplus b_2$, the provers
win the game if $a_c = b_r$.

If the provers follow a classical strategy, their maximum winning probability in
this game is $\frac{8}{9}$, while we describe now a quantum strategy that makes
them win with probability $1$. The provers share
two EPR pairs and, on question $r$ (resp. $c$), the prover performs the measurements
indicated in the first two columns (resp. rows) of row $r$ (resp. column
$c$) of the following table
\begin{table}[H]
  \center
\begin{tabular}{|c|c|c|}
\hline
$IZ$ & $ZI$ & $ZZ$ \\
\hline
$XI$ & $IX$ & $XX$ \\
\hline
$XZ$ & $ZX$ & $YY$\\
\hline
\end{tabular}
\end{table}
\noindent and answer with the outcomes of the measurements. The values $a_3$ and
$b_3$ should correspond to the measurement of the EPR pairs according to the
third column and row, respectively.

The self-testing theorem proved by Wu, Bancal and Scarani~\cite{WuBMS16}  
states  that if the provers win the Magic Square game with
probability close to $1$, they share two EPR pairs and the measurements
performed are close to the honest Pauli measurements, up to local isometries.

\begin{lemma}\label{lem:ms-rigid}
Suppose a strategy for the provers, using state $\ket{\psi}$ and observables
  $W$, succeeds with probability at least $1-\eps$ in the Magic Square game.
  Then there exist  isometries $V_D:\mH_\reg{D} \to (\complex^2\otimes \complex^2)_{\reg{D'}}\otimes {\mH}_{\hat{\reg{D}}}$, for $D\in\{A,B\}$ and a state $\ket{\aux}_{\hat{\reg{A}}\hat{\reg{B}}} \in {\mH}_{\hat{\reg{A}}}\otimes {\mH}_{\hat{\reg{B}}}$ such that
$$\big\| (V_A \otimes V_B) \ket{\psi}_{\reg{AB}} - \ket{\EPR}_{\reg{A}'\reg{B}'}^{\otimes 2} \ket{\aux}_{\hat{\reg{A}}\hat{\reg{B}}} \big\|^2 = O(\sqrt{\eps}),$$
and for $W\in \{I,X,Z\}^2$,
\begin{align*}
\big\| \big(W -V_A^\dagger \sigma_W V_A\big) \otimes \uI_B \ket{\psi} \big\|^2 &= O(\sqrt{\eps}).
\end{align*}
\end{lemma}

\subsubsection{Pauli Braiding Test.}
The starting point of our work is the Pauli Braiding
Test(\pbt)~\cite{NatarajanV16}, 
a non-local game that 
allows the verifier to certify
that two provers share $t$ EPR
pairs and perform the indicated
measurements, which consist of tensors of Pauli observables.

We define \pbt{} in details later in this section and here we state the main properties
that will be used in our Hamiltonian game.
In \pbt{},   each prover receives questions in the form $W \in \{X, Z, I\}^t$,
and each one is answered with some $b \in \{-1, +1\}^t$.  For $W \in \{X, Z\}^t$
and $a\in\{0,1\}^t$, we have $W(a) \in \{X, Z, I \}^t$ where  $W(a)_i = W_i$ if
$a_i = 1$ and $W(a)_i = I$ otherwise.  

In the honest strategy, the provers
share $t$ EPR pairs and measure them with respect to the
observable $\sigma_{W} \eqdef \bigotimes_{i \in [t]} \sigma_{W_i}$ on question
$W$.
However, the provers could deviate and perform an arbitrary strategy, 
sharing an entangled state $\ket{\psi}_{AB} \in \mH_\reg{A} \otimes
\mH_\reg{B}$ and performing projective measurements
 $\tau^{A}_W$ and $\tau^{B}_W$ for each possible question $W$.
It was shown that if the provers pass \pbt{} with probability $1-\eps$,
their strategy is, up to local isometries, $O(\sqrt\eps)$-close to sharing $t$ EPR pairs and measuring
$\sigma_{W}$ on question $W$~\cite{NatarajanV16}. \\

We  describe now \pbt{}.
The test is divided in three different tests,
which are performed equal probability. The first one, the Consistency Test,
checks if the measurement performed by  both provers on question $W$ are
equivalent, i.e. $\tau^{A}_W \otimes I_B\ket{\psi}_{AB}$  is close to $I_A
\otimes \tau^{B}_W \ket{\psi}_{AB}$.  In the Linearity Test, the verifier checks
if the measurement performed by the provers are linear, i.e.
$\tau^{A}_{W(a)}\tau^{A}_{W(a')} \otimes I_B\ket{\psi}_{AB}$ is close to
$\tau^{A}_{W(a+a')} \otimes I_B\ket{\psi}_{AB}$. Finally, in the
Anti-commutation Test, the verifier  checks if the provers' measurements follow
commutation/anti-commutation rules consistent with the honest measurements, 
namely  $\tau^{A}_{W(a)}\tau^{A}_{W'(a')}
\otimes I_B\ket{\psi}_{AB}$ is close to $(-1)^{|\{W_i \ne W'_i \textrm{ and
} a_i = a'_i = 1\}|}\tau^{A}_{W'(a')}\tau^{A}_{W(a)} \otimes
I_B\ket{\psi}_{AB}$. 

The Consistency Test and Linearity Test are very simple and are described in
\Cref{fig:pbt}.  For the Anti-commutation Test, we can use non-local games
that allow the verifier to check that the
provers share a constant number of EPR pairs and perform Pauli measurements on them.  In
this work we use the Magic Square game since there is a perfect
quantum strategy for it.

\begin{figure}[H]
\rule[1ex]{16.5cm}{0.5pt}
\vspace{-20pt}
\justify
The verifier performs the following steps, with probability $\frac{1}{3}$ each:
  \begin{enumerate}[label=(\Alph*)]
  \item Consistency test
    \begin{enumerate}[label=\arabic*.]
    \item The verifier picks $W \inr \{X, Z\}^n$ and $a \in \{0,1\}^n$.
    \item The verifier sends $W(a)$ to both provers. 
      \item The verifier accepts iff the provers' answers are equal.
    \end{enumerate}
  \item Linearity test
    \begin{enumerate}[label=\arabic*.]
      \item The verifier picks $W \inr \{X, Z\}^t$ and $a, a' \inr \binary^t$.
      \item The verifier sends $(W(a),W(a'))$ to $P_1$ and $W' \inr
        \{W(a), W(a')\}$ to $P_2$.
      \item
        The verifier receives $b, b' \in \{\pm1\}^t$ from $P_1$ and 
        $c \in \{\pm1\}^t$ from $P_2$.
      \item 
        The verifier accepts iff $b = c$ when $W' = W(a)$ or $b' =
        c$ when $W' = W(a')$.
    \end{enumerate}
  \item Anti-commutation test
    \begin{enumerate}[label=\arabic*.]
    \item  The verifier makes the provers play Magic Square games in parallel with the
      $t$ EPR pairs (see \Cref{sec:magic-square}).
    \end{enumerate}
\end{enumerate}
\rule[2ex]{16.5cm}{0.5pt}\vspace{-.5cm}
  \caption{Pauli Braiding Test}\label{fig:pbt}
\end{figure}

\begin{theorem}[Theorem 14 of \cite{NatarajanV16}]\label{lem:rigidity}
  Suppose $\ket{\psi}\in\mH_\reg{A}\otimes \mH_\reg{B}$ and $W(a) \in
  \Obs(\mH_\reg{A})$, for $W\in \{X,Z\}^t$ and $a\in\{0,1\}^t$, specify a
  strategy for the players that has success probability at least $1-\eps$ in the
  Pauli Braiding Test. 
  Then there exist isometries $V_D:\mH_\reg{D} \to ((\complex^2)^{\otimes
  t})_{\reg{D'}}  \otimes \hat{\mH}_{\hat{\reg{D}}}$, for $D\in\{A,B\}$, such
  that
  $$\big\| (V_A \otimes V_B) \ket{\psi}_{\reg{AB}} -
  \ket{\EPR}_{\reg{A}'\reg{B}'}^{\otimes t}
  \ket{\aux}_{\hat{\reg{A}}\hat{\reg{B}}} \big\|^2 = O(\sqrt{\eps}),$$
  and on expectation over  $W\in \{X,Z\}^t$,
  \begin{align*}
     \Es{a\in\{0,1\}^t} \big\| \big(W(a) -V_A^\dagger (\sigma_W(a) \otimes
     \uI) V_A\big) \otimes \uI_B \ket{\psi} \big\|^2 &= O(\sqrt{\eps}).
  \end{align*}
  
  Moreover, if the provers share $\ket{\EPR}_{\reg{A}'\reg{B}'}^{\otimes
  t}$ and measure with the observables $\bigotimes \sigma_{W_i}$ on question $W$, they pass the test with probability $1$.
\end{theorem}

\subsection{Local Hamiltonian problem}

The Local Hamiltonian problem can be seen as the quantum analog of MAX-SAT
problem. An instance for this problem consists 
in
$m$
Hermitian matrices $H_1,\ldots,H_m$, where each $H_i$ acts non-trivially on at most 
most $k$ qubits.
For some parameters $\alpha, \beta \in \real$, $\alpha < \beta$, the Local
Hamiltonian problem asks if there is a global state such that its energy  in respect of $H = \frac{1}{m} \sum_{i \in[m]} H_i$ is
at most $\alpha$ or all states have energy at least $\beta$. This problem was
first proved to be \QMA{}-complete for $k = 5$ and $\beta - \alpha \geq
\frac{1}{\poly(n)}$~\cite{KitaevSV02}.
In this work, we are particularly interested in the version of LH where all the
terms are tensor products of $\sigma_X$, $\sigma_Z$ and $\sigma_I$.

\begin{definition}[XZ Local Hamiltonian]\label{def:xz-local-hamiltonian}
  The XZ $k$-Local Hamiltonian problem, for $k \in \Zplus$ and parameters $\alpha, \beta \in [0,1]$
  with $\alpha<\beta$, is the following promise problem. Let $n$ be the number of qubits of a quantum system.
   The input is a sequence of $m(n)$
 values $\gamma_1,...,\gamma_{m(n)}\in [-1,1]$ and $m(n)$
 Hamiltonians $H_1, \ldots, H_{m(n)}$ 
  where $m$ is a polynomial in $n$, and for each $i\in [m(n)]$,
 $H_i$ is of the form $\bigotimes_{j \in n} \sigma_{W_j} \in  \{\oX, \oZ, \oI\}^{\otimes n}$ with
 $|\{j |j \in [n] \textrm{ and } \sigma_{W_j} \ne \oI \}| \leq k$.
  For $H \eqdef \frac{1}{m(n)} \sum_{j = 1}^{m(n)} \gamma_j H_j$, one of the following two conditions hold.
\begin{description}
  \item[\quad Yes.]
{There exists a
      state $\ket{\psi} \in \complex^{2^{n}}$ such that
      $\bra{\psi} H \ket{\psi}
        \leq \alpha(n)$}
  \item[\quad No.]
{For all states $\ket{\psi} \in \complex^{2^{n}}$
      it holds that
      $\bra{\psi} H \ket{\psi}
        \geq \beta(n) .$
        }
\end{description} 
\end{definition} 
~\\
Whenever the value of $n$ is clear from the context, we call $\alpha(n)$, $\beta(n)$ and $m(n)$ by
$\alpha$, $\beta$ and $m$. The $XZ$  $k$-LH problem has been also proved
\QMA{}-complete~\cite{CubittM14,Ji16}.

\begin{lemma}[Lemma $22$ of \cite{Ji16}, Lemma $22$ of \cite{CubittM14}] \label{lem:xz-hamiltonian}
  There exist $\alpha,\beta \in [0,1]$ satisfying $\beta - \alpha \geq
  \frac{1}{poly(n)}$ such that $XZ$ $k$-Local Hamiltonian is \QMA{}-complete,
  for some constant $k$.
\end{lemma}

It is an open question if $k$-LH  is \QMA{}-complete for $\beta - \alpha =
O(1)$ while maintaining $k$  constant~\cite{AharonovAV13}. However,
it is possible to achieve this gap at the cost of increasing the locality of the Hamiltonian \cite{NatarajanV16}. 

\begin{lemma}[Lemma $26$ of \cite{NatarajanV16}]\label{lem:amplification}
Let $H$ be an $n$-qubit Hamiltonian with minimum energy $\lambda_{0}(H) \geq 0$
  and such that $\norm{H} \leq 1$. Let $\alpha, \beta \geq \frac{1}{poly(n)}$
  and $\alpha < \beta$ for all $n$. Let $H'$ be the following Hamiltonian on
  $(\beta
  - \alpha)^{-1}n$ qubits
\[
  H' = \oI ^{\otimes n a} - (\oI^{\otimes n} - (H - a^{-1}\oI^{\otimes n}))^{\otimes a}, \textrm{ where } a = \left(\beta - \alpha \right)^{-1}.
\]

It follows that if $\lambda_0(H) \leq \alpha$ then 
$\lambda_0(H') \leq \frac{1}{2}$, while if  $\lambda_0(H) \geq \beta$ then  $\lambda_0(H') \geq 1$.
Moreover if $H$ is a $XZ$ Hamiltonian, so is  $H'$.
\end{lemma}

Finally, we define now non-local games for Local Hamiltonian problems.

\defclass{Non-local games for Hamiltonians}{def:game-lh}{
A non-local game for the Local Hamiltonian problem consists in a reduction
from a Hamiltonian $H$ acting on $n$ qubits to a non-local game $G(H)$ where a
verifier plays against $r$ provers, and for some parameters $\alpha, \beta,
c,
s \in [0,1]$, for $\alpha < \beta$ and $c > s$, the following holds.
}{
  If $\lambda_0(H) \leq \alpha$, then $\omega^*(G(H)) \geq c$
}{
  If $\lambda_0(H) \geq \beta$, then $\omega^*(G(H)) \leq s$.
}

\section{One-round two-prover game for Local Hamiltonian}
\label{sec:game-lh}
In this section, we define our non-local game for Local Hamiltonian problem,
proving \Cref{thm:main}.
We start with a  $XZ$ Hamiltonian $H =\frac{1}{m} \sum_{l \in m} \gamma_l H_l$
acting on $n$ qubits
and  $\alpha, \beta \in [0,1]$ with $\alpha < \beta$.
We 
propose then the Hamiltonian Test $G(H)$, a non-local game based on $H$,  whose
maximum acceptance probability upper and lower bounds are tightly related to
$\lambda_0(H)$.
Based on $G(H)$, we show how to construct another non-local game
$\tilde{G}(H)$ for which there exists some universal constant $\Delta > 0$
such  that if $\lambda_0(H) \leq \alpha$, then 
$\omega^*(\tilde{G}(H)) \geq \frac{1}{2} + \Delta$, whereas if $\lambda_0(H)
\geq \beta$, then $\omega^*(\tilde{G}(H)) \leq\frac{1}{2} - \Delta$. The
techniques used to devise $G(H)$ and $\tilde{G}(H)$ are based on Ref. \cite{Ji16,NatarajanV16}.\\

We describe now the Hamiltonian Test $G(H)$, which is composed by the Pauli Braiding
Test (\pbt) (see \Cref{sec:pbt}) and the Energy Test (\energy{}), which allows
the verifier estimate $\lambda_0(H)$.  The provers are
expected to share $t$ EPR pairs and the first prover holds a copy of the
groundstate of $H$.  In \et{}, the verifier picks $l \inr [m]$, $W
\inr \{X, Z\}^t$ and $e \inr \{0,1\}^t$, and chooses $\T_1,...,\T_n \in [t]$
such that $W(e)_{\T_i}$ matches the $i$-th Pauli observable of $H_l$.  By
setting $t = O(n\log{n})$, it is possible to choose such positions for a random
$W(e)$ with overwhelming probability. The verifier sends $\T_1,...,\T_n$ to the first prover, who
is supposed to teleport the groundstate of $H$ through the EPR pairs in these positions. As in \pbt{}, the
verifier sends $W(e)$ to the second prover, who is supposed to measure his EPR
halves with the corresponding observables.  The values of $\T_1,...,\T_n$ were chosen
in a way that the first prover teleports the groundstate of $H$ in the exact positions
of the measurement according to $H_l$.
With the outcomes of the teleportation measurements, the verifier can correct
the output of the measurement of the second prover and estimate $\lambda_0(H)$.
The full description of the game is presented in \cref{fig:game}.

\begin{figure}[t]
\rule[1ex]{16.5cm}{0.5pt}
\vspace{-20pt}
\justify
The verifier performs each of the following steps with probability $1 - p$  and $p$, respectively:
  \begin{enumerate}[label=(\Alph*)]
  \item Pauli Braiding Test
  \item Energy Test
    \begin{enumerate}[label=\arabic*.]
\item The verifier picks  $W \inr \{X, Z\}^t$, $e \inr \binary^t$ and  $l \inr [m]$
\item The verifier picks positions $\T_1$, ...$\T_n$ such that $H_l = \bigotimes \sigma_{W(e)_{\T_i}}$.
\item The verifier sends $\T_1, ..., \T_n$ to the first prover and $W(e)$ to the second prover.
\item The first prover answers with $a,b \in \{0,1\}^n$ and the second prover
  with  $c \in \{+1,-1\}^t$.
\item 
Let $d \in \{-1,+1\}^n$ such that $d_i = (-1)^{a_i} c_{\T_i}$ if $W_{\T_i} = X$ and $d_i = (-1)^{b_i} c_{\T_i}$ if $W_{\T_i} = Z$.
\item If $\prod_{i \in [n]} d_i  \ne \fn{sign}{\gamma_l}$, the verifier accepts.
\item Otherwise, the verifier rejects with probability $|\gamma_l|$.
\end{enumerate}
\end{enumerate}
\rule[2ex]{16.5cm}{0.5pt}\vspace{-.5cm}
\caption{Hamiltonian Test $G(H)$ for a $XZ$ Hamiltonian $H$.}\label{fig:game}
\end{figure}

We state now two auxiliary lemmas with lower and upper bounds on the maximum acceptance probability on
$G(H)$.

\begin{lemma}\label{lem:honest-strategy}
Let $H = \sum_{l \in [m]} \gamma_l H_l$ be a $XZ$ Hamiltonian, let $G(H)$ be
  the Hamiltonian-self test for $H$,  described in \Cref{fig:game}, and 
\[\omega_h(H) \eqdef 1 - p\left(\frac{1}{2m} \sum_{l \in [m]}|\gamma_l| -
\frac{1}{2}\lambda_0(H)\right). \]
  If the
  provers use the honest strategy in \pbt{}, the maximum acceptance
  probability in $G(H)$ is  $\omega_h(H)$. Moreover, this
  probability is achieved if the first prover behaves honestly in \et{}.
\end{lemma}
\begin{lemma}\label{lem:soundness-game}
  Let $H$, $G(H)$ and $\omega_h(H)$ be defined as \Cref{lem:honest-strategy}.
For every $\eta > 0$, there is some value of $p = O(\sqrt \eta)$ such that 
  $\omega^*(G(H)) \leq \omega_h(H)
  + \eta$.
\end{lemma}

 We defer the proof of these lemmas
 to \Cref{sec:proof} and we concentrate now in
 proving our main theorem.

\begin{theorem}\label{thm:main}
There exists a universal constant $\Delta$ such that the following holds.
  Let $H = \sum_{l \in m} \gamma_l H_l$ be $XZ$ $k$-Local Hamiltonian
  acting on  $n$ qubits with parameters   $\alpha, \beta \in (0,1)$, for  $\beta
  > \alpha$. There exists one-round two-prover non-local game such that
\begin{itemize}
  \item if $\lambda_0(H) \leq \alpha$, then the verifier accepts with probability at least $\frac{1}{2} + \Delta$; and
  \item if  $\lambda_0(H) \geq \beta$,
    then the verifier accepts with probability at most $\frac{1}{2} - \Delta$.
\end{itemize}
  Moreover, each message is $\tilde{O}(n(\beta - \alpha)^{-1})$-bit long.
\end{theorem}
\begin{proof}
\Cref{lem:amplification} states that from $H$ we can construct a Hamiltonian
  $H'$ such  that
\[\lambda_0(H)\leq \alpha \Rightarrow \lambda_0(H')\leq \frac{1}{2} \textrm{ and } \lambda_0(H)\geq \beta \Rightarrow \lambda_0(H') \geq 1,\]
  and $H' = \sum_{l \in m'} \gamma'_l H'_l$ is an instance of $XZ$ Local
  Hamiltonian problem.

We now bound the maximum acceptance probability of the  Hamiltonian Test on
  $H'$, relating it to the groundstate energy of $H$.
  From \Cref{lem:honest-strategy} 
  it follows that 
  \[\lambda_0(H)\leq \alpha \Rightarrow \omega^*(G(H')) \geq 
  1 - p\left(\frac{1}{2m} \sum_{l \in [m]}|\gamma'_l| - \frac{1}{4}\right) \eqdef c,\]
while from \Cref{lem:soundness-game}, for any $\eta > 0$ and some $p \leq C \sqrt{\eta}$, we have that
  \[
\lambda_0(H)\geq \beta \Rightarrow
\omega^*(G(H')) \leq 1 - p\left(\frac{1}{2m} \sum_{l \in [m]}|\gamma'_l| -
  \frac{1}{2}\right) + \eta= 
 c - \frac{C \sqrt\eta}{4} + \eta.\]

By choosing  $\eta$ to be a constant such that  
$\eta' \eqdef \frac{C \sqrt\eta}{4} - \eta > 0$, it follows that
  \[\lambda_0(H)\leq \alpha \Rightarrow \omega^*(G(H')) \geq 
  c \text{ and } 
  \lambda_0(H)\geq \beta \Rightarrow \omega^*(G(H')) \leq 
  c - \eta'.  \]
  
We describe now the game $\tilde{G}(H)$ that achieves the completeness and
  soundness properties  stated in the theorem. In this game, the verifier
  accepts with probability $\frac{1}{2} - \frac{2c - \eta'}{4}$, rejects with
  probability $\frac{2c - \eta'}{4}$ or play $G(H')$ with probability
  $\frac{1}{2}$.  Within this new game, if $\lambda_0(H)\leq \alpha$ then
  $\omega^*(\tilde{G}(H')) \geq \frac{1}{2} + \frac{\eta'}{4}$, whereas when
  $\lambda_0(H)\geq \beta$, we have that $\omega^*(\tilde{G}(H')) \leq \frac{1}{2} -
  \frac{\eta'}{4}$.
\end{proof}

\begin{corollary}
There  exists  a  protocol for  verifiable  delegation  of  quantum  computation
where a classical client communicates with two entangled servers in one round of
classical communication.
\end{corollary}
\begin{proof}
The corollary holds from composing the circuit-to-hamiltonian construction (see \Cref{sec:circ-to-ham}) with our non-local game.
\end{proof}

\begin{remark}
The parameters of our delegation protocol allow us to use standard arguments in relativistic cryptography to replace the assumption that the provers do not communicate by the assumption that they can only communicate at most as fast as the speed of light. See \Cref{sec:applications} for more details on this.
\end{remark}

\subsection{Proof of \Cref{lem:honest-strategy,lem:soundness-game}}
\label{sec:proof}
We start by proving \Cref{lem:honest-strategy}, showing an upper bound on the
acceptance probability if the provers are honest in \pbt{}.
\begin{proof}[Proof of \Cref{lem:honest-strategy}]
  Since \pbt{} and \et{} are indistinguishable to the second prover, he also
  follows the honest strategy in \et{} and 
  the acceptance probability in $G(H)$ depends uniquely in the strategy of the
  first prover in \energy{}.
  
  Let $a, b \in \{0,1\}^n$ be the answers of the first prover in $\et$ and
$\tau$ be the reduced state held by the second prover on the positions
  $\T_1,...,\T_n$ of his EPR halves, after the teleportation.

For a fixed $H_l$, the verifier rejects with
  probability
  \begin{align}
    \label{eq:rej-prob}
    \frac{|\gamma_l| + \gamma_l \mathbb{E}\left[\prod_{i \in n} d_i\right] }{2}.
  \end{align}
  We notice that measuring a qubit $\ket{\phi}$ in the $Z$-basis with outcome
  $f \in \{\pm 1\}$ is equivalent of considering the outcome $(-1)^g f$ when
  measuring $X^gZ^h\ket{\phi}$ in the same basis. An
  analog argument follows also for the $X$-basis. Therefore, by fixing the
  answers of the first prover,  instead of considering that
  the second prover measured  $\tau$ in respect of $H_l$  with outcome $c$, we
  consider that  he measured $\rho = \uZ^b \uX^a \tau \uX^a \uZ^b$ with respect
  to 
  $H_l$ with outcome $d$. In this case, 
  by taking $\prod_{i \in n} d_i$ as the outcome of the measurement of $H_l$ on
  $\rho$, and
  averaging over all $l \in [m]$, it follows from \cref{eq:rej-prob} that the
  verifier rejects in \energy{} with probability
\begin{equation*}\label{eq:average}
\frac{1}{m} \sum_{l \in [m]} \frac{|\gamma_l| + \gamma_l\trace{\rho H_l}}{2} = 
\frac{1}{2m}\sum_{l \in [m]}|\gamma_l| + \frac{1}{2} \trace{\rho H},
\end{equation*}
and this value is minimized when $\rho$ is the groundstate of $H$. In this case
  the overall acceptance probability in $G(H)$ is at most
\begin{equation*}\label{eq:maximum-honest-probability}
  1 -p\left(\frac{1}{2m} \sum_{l \in [m]}|\gamma_l| -
  \frac{1}{2}\lambda_0(H)\right) = \omega_h(H). 
\end{equation*}

Finally, this acceptance probability is achieved  if the first prover teleports
  the groundstate $\ket{\psi}$ of $H$ and report the honest outcomes
  from the teleportation, since $\tau = \uX^a\uZ^b\kb{\psi}\uZ^b\uX^a$ and $\rho
  = \kb{\psi}$.
\end{proof}

We use now the
self-testing of \pbt{} to certify the measurements of the second prover in
\et{}. In this way, we can bound the acceptance probability in $G(H)$ with
\Cref{lem:honest-strategy} and prove \Cref{lem:soundness-game}.

\begin{proof}[Proof of \Cref{lem:soundness-game}]
Let $S$ be the strategy of the provers, which results in
  acceptance
  probabilities 
$1 - \eps$ in $\pbt{}$ and $1- \frac{1}{2m} \sum_{l \in [m]}|\gamma_l| -
  \frac{1}{2}\lambda_0(H) + \delta$ in $\et{}$,
  for some $\eps$ and $\delta$. 

By Lemma~\ref{lem:rigidity}, their strategy in \pbt{} is
  $O(\sqrt\eps)$-close to the honest strategy, up to the local isometries $V_A$ and
  $V_B$. Let $S_h$ be the strategy where the provers follow 
  the honest strategy in \pbt{} and, for \energy{},
  the first prover  performs the same operations of $S$, but considering the
  isometry $V_A$ from \Cref{lem:rigidity}.
Since the measurements performed by the provers in $S$ and $S_h$ are
  $O(\sqrt\eps)$-close to each other, considering the isometries, the
  distributions of the corresponding transcripts have statistical distance
  at most $O(\sqrt\eps)$. Therefore, the provers following strategy $S_h$ are accepted in \energy{} with probability at least \[1- \frac{1}{2m} \sum_{l \in [m]}|\gamma_l| - \frac{1}{2}\lambda_0(H) + \delta - O(\sqrt\eps).\]

  Since in $S_h$ the provers perform the honest strategy in \pbt{}, it follows
  from \Cref{lem:honest-strategy} that
  \[1- \frac{1}{2m} \sum_{l \in [m]}|\gamma_l| - \frac{1}{2}\lambda_0(H) +
  \delta - O(\sqrt\eps) \leq 1 - \frac{1}{2m} \sum_{l \in [m]}|\gamma_l| -
  \frac{1}{2}\lambda_0(H),\]
  which implies that $\delta \leq C \sqrt{\eps}$, for some constant $C$.

The original strategy $S$ leads  to acceptance probability at
  most
\begin{align*}
(1-p)(1-\eps) + p\left(1-\frac{1}{2m} \sum_{l \in [m]}|\gamma_l| - \frac{\lambda_0(H)}{2} +  C \sqrt{\eps}\right)
  = \omega_h(H)  -(1-p)\eps +   p C \sqrt{\eps}.
\end{align*}
  For any $\eta$, we can pick $p = \min\left\{\frac{\sqrt{\eta}}{D}, 1\right\}$,
  for $D \geq 2C$, and it follows that
\begin{align*}
  p C\sqrt\eps- (1-p)\eps
  \leq \frac{2C\sqrt\eta \sqrt\eps}{D} - \eps
  \leq \sqrt\eta \sqrt\eps - \eps
  \leq \eta
\end{align*}
  and therefore the maximum acceptance probability is at most $\omega_h(H) +
  \eta$.
\end{proof}

\clearpage

\bibliography{delegation}

\appendix

\section{Circuit-to-hamiltonian construction}
\label{sec:circ-to-ham}
Feynman~\cite{Feynman1986}, in his pioneering work where he suggests the use of
the quantum structure
of matter as a computational resource, has  shown the construction of a time-independent
Hamiltonian that is able to simulate the evolution of a quantum circuit. 
This construction is now known as the circuit-to-hamiltonian construction and
it is a central point in proving \QMA{}-completeness of Local Hamiltonian
problems \cite{KitaevSV02,BreuckmannT14} and the universality of adiabatic
quantum computation \cite{AharonovDKLLR04}. More related to this work, it has
also been used in the context of delegation of quantum
computation~\cite{FitzsimonsH15,NatarajanV16}. We describe now such
construction.

Let $Q = U_T...U_1$ be a quantum circuit on $n$ qubits, decomposed on $T$
$2$-qubit gates. Let us assume for simplicity that the circuit $Q$ is applied
on the initial state $\ket{\psi}$. The Hamiltonian $H_Q$ acts 
on $n$ working qubits, as the circuit $Q$, and an extra
clock-register of $c$ qubits to count the operations steps from $0$ to $T$. The
number of bits in the clock register depends on the representation of 
the time steps:
if we represent it in binary, we take $c = \log{T}$; for some applications, it
is better to represent it in unary, where time $t$ will be encoded as $T-t$ ``$0$''s followed by $t$
``$1$''s. For the remainder of the section we abstract the representation of the
clock register and we write
$\ket{t}_{clock}$ for the correct encoding of time $t$.
We will construct $H_Q$ such that its groundstate is
\begin{align}
  \label{eq:history}
  \sum_{t = 0}^{T} \ket{t}_{clock} \otimes U_t...U_1\ket{\psi},
\end{align}
which is known as the history state of $Q$.
As noticed by Fitzsimons and Hajdu\v{s}ek~\cite{FitzsimonsH15}, 
the history state of $H_Q$ can be computed
in quantum polynomial time if the initial state $\ket{\psi}$ is provided. 
\begin{lemma}
  \label{lem:efficient-history}
  Provided the initial state $\ket{\psi}$ of $Q$, the history state
  $\sum_{t = 0}^{T} \ket{t}_{clock} \otimes U_t...U_1\ket{\psi}$
  can be prepared in time polynomial in $T$.
\end{lemma}

The Hamiltonian $H_Q$ is decomposed in three parts:
the initialization terms, 
the propagation terms and clock terms. As we see later, output
terms are also needed for some applications.
  
The initialization terms check if the groundstate is a computation that start in a valid
initial state $\ket{\psi}$. For instance, if $\ket{\psi} = \ket{0}^{\otimes n}$, then for
each $i \in [n]$, the following term will be added to $H_Q$
\[\kb{0}_{clock} \otimes \kb{1}_i.\]
The interpretation of these terms is that they add some ``penalty'' for states
where the computation does not start with a $\ket{0}^{\otimes n}$.

The propagation terms check if  all the intermediate
steps $U_0$,...,$U_T$ are simulated in the Hamiltonian. For
each step $t \in [T]$, the following Hamiltonian is added to $H_Q$ 
\begin{align*}
    \frac{1}{2} \Big(&-\ket{t}\bra{t-1}_{clock} \otimes U_t
 - \ket{t}\bra{t-1}_{clock} \otimes U_t^\dagger  \\
 &+  \kb{t}_{clock} \otimes I  + 
\kb{t-1}_{clock}  \otimes I \Big),
\end{align*}
where the second part of the tensor product acts on the same qubits of  $U_t$.

The clock terms are added in order to check if the clock register
contains only correct encodings of time.
For instance, if time is encoded in unary, 
the clock terms  check if there is no $1$ followed by a $0$ in
the clock register. More concretely, in the unary representation, 
for every $i \in [T]$, the following
term is added to $H_Q$
\[\kb{10}_{i, i+1},\]
and it acts on qubits $i$ and $i +1$ of the clock register

One can easily see by inspection that the state in \cref{eq:history} is the only
state that has energy $0$ according the previous terms.

In some applications, we need also to check some properties of the output of
$Q$. For instance in delegation protocols, we are interested in the probability that $Q$
outputs $\ket{1}$ (we usually say in this case that $Q$ accepts).
In these cases, we want to construct $H_Q$ such that its frustration is related
to the acceptance probability of the circuit: if $Q$ outputs $\ket{1}$ with
probability at least $c$, then $\lambda_0(H_Q) \leq \alpha$,
while if $Q$ outputs $\ket{1}$ with
probability at most $s$, then $\lambda_0(H_Q) \geq \beta$.
For this task, we add the following term to the $H_Q$ that acts on the
clock register and on the output qubit
\[\kb{T}_{clock} \otimes \kb{0}_{output}.\]

The following theorem was then proved by Kitaev~\cite{KitaevSV02}.
\begin{theorem}[Sections $14.4.3$ and $14.4.4$ of \cite{KitaevSV02}]
  \label{thm:kitaev}
  Let $Q$ be a quantum circuit composed by $T$ gates that computes on some
  initial state $\ket{\phi}$ and then decides to accept or
  reject. Let $H_Q$ be the $5$-Local Hamiltonian created with the circuit-to-hamiltonian
  with unary clock on $Q$.
  \begin{description}
    \item[Completeness.] If the acceptance probability is at least $1-\eps$, then
      $\lambda_0(H_Q) \leq \frac{\eps}{T+1}$.
    \item[Soundness.] If the acceptance probability is at most $\eps$, then
      $\lambda_0(H_Q) \geq c\frac{1-\sqrt\eps}{T^3}$, for some constant $c$.
  \end{description}
\end{theorem}

Ji~\cite{Ji16} has proved that Kitaev's construction can be converted into an XZ
Local Hamiltonian, by choosing a suitable gate-set for the circuit $Q$.

\begin{theorem}[Lemma $22$ of \cite{Ji16}]
  Let $Q$ be a quantum circuit composed of gates in the following universal
  gate-set $\{CNOT, X, \cos(\frac{\pi}{8})X + \sin(\frac{\pi}{8})Z\}$. 
  Then $H_Q$ from~\cref{thm:kitaev} can be written as a $XZ$ $5$-Local
  Hamiltonian.
\end{theorem}

\section{Relativistic cryptography}
\label{sec:applications}

In this section we present the main arguments to show that in our protocol, the
assumption that the provers do not communicate can be replaced by the assumption
that they cannot communicate faster than speed of light. Such protocols are
called relativistic protocols in the literature.

The first relativistic protocol is due to Kent~\cite{Kent99}, who showed the
existence of relativistic information-theoretical secure bit commitment,
which is impossible in the general case~\cite{Mayers97,LoC97}.  Since
then, several other relativistic protocols were proposed for bit
commitment \cite{Kent11,Kent12,LunghiKBHTKGWZ13,LunghiLBJTWZ15,ChakrabortyCL15}, verification of space-time position of
agents~\cite{ChandranGMO2009,KentMS11,LauL11,BuhrmanCFGGOS11,Unruh2014,ChakrabortyL15}, oblivious transfer \cite{Pitalua16} and zero-knowledge proof
systems \cite{ChaillouxL17}.

Relativistic protocols can be achieved by fixing the position of the provers and the verifier
in
such a way that the information transmitted between the provers takes
much longer than an upper-bound of the duration of the honest protocol.  The
verifier could  abort whenever the answers arrive too late,
since the provers could have communicated and the security of the protocol is
compromised.  The space-time diagram of such interactions is depicted in
\Cref{fig:spacetime}.
We show also how to prevent more sophisticated attacks
where malicious provers move closer to the verifier in order to receive the
message earlier, being able to cheat in the previous protocol.

\begin{figure}
\begin{tikzpicture} \pgfmathsetmacro{\pone}{0.05}
  \pgfmathsetmacro{\ptwo}{0.55}
  \pgfmathsetmacro{\ver}{0.3}
  \pgfmathsetmacro{\tzero}{0.25}
  \pgfmathsetmacro{\tone}{0.4}
  \pgfmathsetmacro{\ttwo}{0.65}
  \pgfmathsetmacro{\tthree}{0.75}
  
  \begin{axis}[ymin=-0.1,ymax=0.8,xmax=0.6,xmin=-0.2,xticklabel=\empty,yticklabel=\empty,
               axis lines = middle,xlabel=$x$,ylabel=$t$,label style =
               {at={(ticklabel cs:1.1)}},
               xtick={\pone, \ptwo, \ver}, ytick={\tzero, \tone, \ttwo, \tthree}]
    \addplot[thick,color=blue,opacity=0,mark=o,fill=green, 
                    fill opacity=0.2]coordinates  {
            (0, 0) 
            (0, 0.75)
            (1, 0.75)
            (1, 0)  };
\addplot [dotted] coordinates {(\pone, 0) (\pone, 1)};
\addplot [dotted] coordinates {(\ptwo, 0) (\ptwo, 1)};
\addplot [dashed] coordinates {(\ver, 0) (\ver, 1)};
\addplot [thick] coordinates {(\ver, 0) (\ptwo, \tzero)};
\addplot [thick] coordinates {(\ver, 0) (\pone, \tzero)};
\addplot [red, thick] coordinates {(\pone, \tzero) (\ptwo, \tthree)};
\addplot [thick] coordinates {(\pone, \tzero) (\pone, \tone)};
\addplot [thick] coordinates {(\ptwo, \tzero) (\ptwo, \tone)};
\addplot [thick] coordinates {(\pone, \tone) (\ver, \ttwo)};
\addplot [thick] coordinates {(\ptwo, \tone) (\ver, \ttwo)};

   \node at (axis cs: \ver, -0.05) {\footnotesize$V$}; 
   \node at (axis cs: \pone, -0.05) {\footnotesize$P_1$}; 
   \node at (axis cs: \ptwo, -0.05) {\footnotesize$P_2$}; 
   \node at (axis cs: -0.05, \tzero) {\footnotesize $t_0$}; 
   \node at (axis cs: -0.1, \tone) {\footnotesize $t_0+t_1$}; 
   \node at (axis cs: -0.1, \tthree) {\footnotesize$3t_0$}; 
   \node at (axis cs: -0.1, \ttwo) {\footnotesize$2t_0 + t_1$}; 
  \end{axis}
\end{tikzpicture}
  \caption{Space-time diagram for the one-round two-prover delegation protocol. The dotted lines correspond to the
  position of the two provers and the dashed
  line corresponds to the
  position of the verifier. The solid black line corresponds to the messages
  exchanged during the honest interaction of the protocol: the verifier sends
  a message to each prover which arrives at time $t_0$; the provers
  perform some computation whose time complexity is upper-bounded by $t_1 \ll
  t_0$; and finally the provers answer back to the
  verifier. The red line
  corresponds to a malicious prover that sends a message to the other prover as soon as
  the verifier's message arrives. The green area corresponds to the period in time that the
  verifier has the guarantee that the provers have not communicated, assuming NSS.}
  \label{fig:spacetime}
\end{figure}

\begin{lemma}
  \label{lem:relativistic}
  Every classical one-round  protocol for delegation of quantum computation with two non-communicating provers can be
  converted into a delegation protocol where the provers are allowed to communicate at
  most as fast as speed of light.
\end{lemma}
\begin{proof}
Let us assume that the provers can be forced to stay at an arbitrary position in
  space.
  We use the unit system, where the speed of light is $c=1$, in order to
  simplify the equations. The provers and the verifier
  are placed in a line, with the verifier at position $0$, 
  the first prover at position $-t_0$ and the second
  prover at position $ t_0$. The value $t_0$ is chosen to be much larger than
  an upper-bound of the time complexity of the provers in the honest protocol, denoted by  $t_1$.
  
  The message from the verifier to the provers arrive at time $t_0$. The
  expected time for the provers to perform the computation and answer back is
  $t_1 + t_0$, whereas the time it takes for the provers to send a
  message to each other is $2t_0 \gg t_1 + t_0$. Therefore, the
  verifier aborts whenever the provers' answers arrive after time $3t_0$ since
  the security of the protocol is compromised. We depict this protocol in
  \cref{fig:spacetime}.

  The previous argument works if we can rely on the position of the provers, but
  in some settings we require the protocol to be robust against malicious
  provers that may move in order to receive the verifier's messages earlier,
  being able  to collude and break the security of the protocol. This
  attack is depicted in \cref{fig:attack}.  We can prevent this type of attacks
  by adding two trusted agents at the expected position of the provers.  The
  verifier sends the message to the agents through a secure channel and the
  agents transmit the information to the provers. The provers perform their
  computation and then report their answers to the agents, who transmit the messages
  to the verifier. The secure channel can be
  implemented with the verifier and the agents sharing one-time pad keys, and
  then messages can be exchanged in a  perfectly secure way. This protocol is
  depicted in \cref{fig:solve-attack}. 
\end{proof}

The fact that the provers communicate after the verifier receives their responses
is not harmful since this cannot change the output of the protocol.

\begin{figure}[h]
~\\
~\\
~\\
~\\
~\\
~\\
~\\
~\\
~\\
~\\
~\\
~\\
\begin{subfigure}{0.4\textwidth}
\begin{tikzpicture}[scale = 0.9, transform canvas={scale=0.9}]
  \pgfmathsetmacro{\pone}{0.05}
  \pgfmathsetmacro{\ponep}{0.25}
  \pgfmathsetmacro{\ptwo}{0.55}
  \pgfmathsetmacro{\ptwored}{0.546}
  \pgfmathsetmacro{\ver}{0.3}
  \pgfmathsetmacro{\tzero}{0.25}
  \pgfmathsetmacro{\tone}{0.4}
  \pgfmathsetmacro{\tonered}{0.396}
  \pgfmathsetmacro{\ttwo}{0.65}
  \pgfmathsetmacro{\ttwored}{0.646}
  \pgfmathsetmacro{\tthree}{0.75}
  \pgfmathsetmacro{\tattack}{0.05}
  \pgfmathsetmacro{\tattacktwo}{0.35}
  
  \begin{axis}[ymin=-0.1,ymax=0.8,xmax=0.6,xmin=-0.2,xticklabel=\empty,yticklabel=\empty,
               axis lines = middle,xlabel=$x$,ylabel=$t$,label style =
               {at={(ticklabel cs:1.1)}},
               xtick={\pone, \ptwo, \ver, \ponep}, ytick={\tzero, \tone, \ttwo}]
\addplot [dotted] coordinates {(\pone, 0) (\pone, 1)};
\addplot [dotted] coordinates {(\ptwo, 0) (\ptwo, 1)};
\addplot [dashed] coordinates {(\ver, 0) (\ver, 1)};
\addplot [thick] coordinates {(\ver, 0) (\ptwo, \tzero)};
\addplot [thick] coordinates {(\ver, 0) (\pone, \tzero)};
\addplot [thick] coordinates {(\pone, \tzero) (\pone, \tone)};
\addplot [thick] coordinates {(\ptwo, \tzero) (\ptwo, \tone)};
\addplot [thick] coordinates {(\pone, \tone) (\ver, \ttwo)};
\addplot [thick] coordinates {(\ptwo, \tone) (\ver, \ttwo)};
\addplot [red, thick] coordinates {(\ponep, \tattack) (\ptwored, \tattacktwo)};
\addplot [red, thick] coordinates {(\ptwored, \tattacktwo) (\ptwored, \tonered)};
\addplot [red, thick] coordinates {(\ptwored, \tonered) (\ver, \ttwored)};

   \node at (axis cs: -0.1, 2) {};
   \node at (axis cs: \ver, -0.05) {\footnotesize$V$}; 
   \node at (axis cs: \pone, -0.05) {\footnotesize$P_1$}; 
    \node at (axis cs: \ponep, -0.05) {\footnotesize \textcolor{red}{$P'_1$}}; 
   \node at (axis cs: \ptwo, -0.05) {\footnotesize$P_2$}; 
   \node at (axis cs: -0.05, \tzero) {\footnotesize $t_0$}; 
   \node at (axis cs: -0.05, \tzero) {\footnotesize $t_0$}; 
   \node at (axis cs: -0.1, \ttwo) {\footnotesize$2t_0 + t_1$}; 
  \end{axis}
\end{tikzpicture}
      \caption{}
      \label{fig:attack}
\end{subfigure}
\quad \quad 
\quad \quad 
\begin{subfigure}{0.4\textwidth}
  \pgfmathsetmacro{\pone}{0.05}
  \pgfmathsetmacro{\ptwo}{0.55}
  \pgfmathsetmacro{\ptwored}{0.546}
  \pgfmathsetmacro{\ver}{0.3}
  \pgfmathsetmacro{\verone}{0.1}
  \pgfmathsetmacro{\vertwo}{0.5}
  \pgfmathsetmacro{\tzero}{0.25}
  \pgfmathsetmacro{\tone}{0.4}
  \pgfmathsetmacro{\tonered}{0.396}
  \pgfmathsetmacro{\ttwo}{0.65}
  \pgfmathsetmacro{\ttwored}{0.646}
  \pgfmathsetmacro{\tthree}{0.75}
  \pgfmathsetmacro{\tattack}{0.2}
  \pgfmathsetmacro{\tattacktwo}{0.6}
  
\begin{tikzpicture}[scale = 0.9, transform canvas={scale=0.9}]
  \begin{axis}[ymin=-0.1,ymax=0.8,xmax=0.6,xmin=-0.2,xticklabel=\empty,yticklabel=\empty,
               axis lines = middle,xlabel=$x$,ylabel=$t$,label style =
               {at={(ticklabel cs:1.1)}},
               xtick={\pone, \ptwo, \ver}, ytick={\tzero, \tone, \ttwo}]
    \addplot[thick,color=blue,opacity=0,mark=o,fill=green, 
                    fill opacity=0.2]coordinates  {
            (0, 0) 
            (0, 0.75)
            (1, 0.75)
            (1, 0)  };
\addplot [dotted] coordinates {(\pone, 0) (\pone, 1)};
\addplot [dotted] coordinates {(\ptwo, 0) (\ptwo, 1)};
\addplot [dashed] coordinates {(\ver, 0) (\ver, 1)};
\addplot [double, thick] coordinates {(\ver, 0) (\ptwo, \tzero)};
\addplot [double, thick] coordinates {(\ver, 0) (\pone, \tzero)};
\addplot [thick] coordinates {(\pone, \tzero) (\pone, \tone)};
\addplot [thick] coordinates {(\ptwo, \tzero) (\ptwo, \tone)};
\addplot [double, thick] coordinates {(\pone, \tone) (\ver, \ttwo)};
\addplot [double, thick] coordinates {(\ptwo, \tone) (\ver, \ttwo)};
\addplot [red, thick] coordinates {(\pone, \tzero) (\ptwo, \tthree)};

   \node at (axis cs: \ver, -0.05) {\footnotesize$V$}; 
    \node at (axis cs: \pone, -0.05) {\footnotesize$P_1, V_1$}; 
    \node at (axis cs: \ptwo, -0.05) {\footnotesize$V_2, P_2$}; 
   \node at (axis cs: -0.05, \tzero) {\footnotesize $t_0$}; 
   \node at (axis cs: -0.05, \tzero) {\footnotesize $t_0$}; 
   \node at (axis cs: -0.1, \ttwo) {\footnotesize$2t_0 + t_1$}; 
  \end{axis}
\end{tikzpicture}
      \caption{}
      \label{fig:solve-attack}
\end{subfigure}
  \caption{
    In \ref{fig:attack}, we illustrate the attack where the first prover changes his position in order to
      receive the message from the verifier earlier,  being able to
      communicate with the second prover in order to affect his answers to the
      verifier. The black lines correspond to the expected protocol, while the
      red ones correspond to the attack.
      In \ref{fig:solve-attack}, we show how to  prevent this attack: the
      verifier places agents in the same position where the provers
      should be. The verifier communicates with these agents through a
      secure channel, and the agents
      communicate with the provers. In this case, we can see that the
      the verifier has again a guarantee that the provers have not communicated
      up to some point in the protocol, here depicted in green.
  }
  \label{fig:spacetime-attacks}
\quad \quad 
\quad \quad 
\end{figure}

\end{document}